\documentclass{article}

\usepackage{amsmath,amsfonts}
\usepackage{algorithmic}
\usepackage{algorithm}
\usepackage{array}
\usepackage[caption=false,font=normalsize,labelfont=sf,textfont=sf]{subfig}
\usepackage{textcomp}
\usepackage{stfloats}

\usepackage{verbatim}

\usepackage{cite}

\newtheorem{lemma}{lemma}
\newtheorem{definition}{Definition}
\newtheorem{proof}{Proof}

\usepackage{PRIMEarxiv}

\usepackage[utf8]{inputenc} 
\usepackage[T1]{fontenc}    
\usepackage{hyperref}       
\usepackage{url}            
\usepackage{booktabs}       
\usepackage{amsfonts}       
\usepackage{nicefrac}       
\usepackage{microtype}      
\usepackage{lipsum}
\usepackage{fancyhdr}       
\usepackage{graphicx}       
\graphicspath{{media/}}     

\title{A Direct Approach for Solving Cloud Computing Task Assignment with Soft Deadlines
}

\author{
  Guang Fang \\
  \texttt{fguang@whu.edu.cn}  \\
  \And
  Yuxiang Zhao \\
  \texttt{yuxiangzhao@whu.edu.cn}
}

\begin{document}
\maketitle

\begin{abstract}
Job scheduling in cloud computing environments is a critical yet complex problem. Cloud computing user job requirements are highly dynamic and uncertain, while cloud computing resources are heterogeneous and constrained. This paper studies the online resource allocation problem for elastic computing jobs with soft deadlines in cloud computing environments. The main contributions include: 1) Integer linear programming modeling is used to design an auction time scheduling framework with three key modules - resource allocation, evaluation, and operation, which can dynamically allocate resources in closed loops. 2) Methods such as time-based single resource utilization evaluation and weighted average evaluation are proposed to evaluate resource usage efficiency. 3) Soft acceptance protocols are introduced to achieve elastic online resource allocation. 4) The time complexity of the proposed algorithms is analyzed and proven to be polynomial time, demonstrating efficiency. 5) Modular design makes the framework extensible. This paper provides a structured cloud computing auction framework as a reference for building practical cloud resource management systems. Future work may explore more complex models of random arrival and multi-dimensional resource constraints, evaluate algorithm performance on real cloud workloads, and further enhance system robustness, efficiency and fairness.
\end{abstract}


\section{Introduction}

With the rapid development of the mobile Internet, smart devices have become an indispensable part of people’s life. Increasingly complex applications such as mobile payment, smart healthcare, mobile games, and virtual reality (VR) put higher requirements on the resource capacity of smart devices. Cloud computing is a cost-effective model that provides abundant applications and services while making information technology (IT) management more accessible and responding to users’ demands faster \cite{Velte2009CloudCA}. The services (computing, communication, storage, and all necessary services) are delivered and implemented in a simplified way: on-demand, regardless of the users’ location and the type of smart devices. In the past decade, two types of cloud platforms blossomed on the Internet, including (i) large scale Internet data centers, exemplified by Amazon EC2 \cite{Amazon2023ec2}, Microsoft Azure and Linode \cite{azure2023,linode2023}, which organize a shared resource pool for serving their users; and (ii) co-location data centers, often found in metropolitan areas, where smaller clouds from different users are physically co-located, jointly managed and serviced by the co-location \cite{Zhang2015Atruthful}.

Virtualization technologies help cloud providers pack their resources into different types of virtual machine (VM), for allocation to cloud users. For example, Amazon EC2 \cite{Amazon2023ec2} currently offers 70 types of different VM types in 22 product categories. Each type of VM has its focus and forte, and a large computing job often requires cooperation among multiple VM instances. For example, social games \cite{Gonalves2023SocialGA} and enterprise applications \cite{Hajjat2010CloudwardBP} are often composed of a front-end web server tier, a load balancing tier and a back-end data storage tier, each suited for execution on a VM that is abundant in a particular type of resource: bandwidth, CPU, or storage. 

Scheduling\cite{Kumar2019ACS} is the process of determining which activities to execute taking into account service quality parameter constraints. Resource scheduling strategies are considered an important role in cloud computing environments. Cloud computing users may need a large number of virtual resources, so it is not possible to manually allocate jobs and resources for all users. Resource allocation studies how to rationally and effectively allocate resources in edge computing systems to accomplish offloading and task processing. In general, the main resources involved in current resource allocation research are computing resources, communication resources, and storage resources. Computing resources usually refer to CPU cycles and resource blocks (VMs/containers). Resource allocation is described in computing as the methods of allocating available resources to requested applications. Virtual machine placement and task scheduling are two key resource allocation problems that can be solved by studying different demands and criteria. Through analysis and summarization of existing studies\cite{Kumar2019ACS,SantoyoGonzlez2020NetworkAwarePO}, virtual machine placement is a proactive way of resource provisioning from the resource provider's perspective. Task scheduling is a passive way of resource provisioning from the user's perspective. The task assignment problem in cloud computing refers to the optimal placement and matching schemes between user tasks and cloud resources.

Cloud computing jobs can be categorized into two types, depending on whether their computational demands are elastic or not. Large web servers and other cloud jobs utilize cloud services as utilities and require rented virtual machines to always be active, and may dynamically scale up and down. These jobs are similar to power consumers on the grid that require an uninterrupted power supply. Other jobs like big data analytics and Google crawl data processing are often batch in nature. They need to complete specific compute jobs without requiring virtual machine services to always be online, and can tolerate some degree of latency as jobs complete. These users are akin to energy consumers on the grid that need to acquire a fixed amount of energy within flexible time windows to power given jobs.

Existing cloud market mechanisms, especially auction-based ones, have implicitly targeted the first type of inelastic cloud jobs. In such one-round \cite{Zhang2014DynamicRP} and online \cite{Shi2014AnOA} cloud resource auctions, once a bid is accepted, the service time window of the corresponding virtual machine is fixed, i.e., the start and end times specified in that round \cite{Shi2014AnOA} or bid \cite{Zhang2015OnlineAI}. These auction algorithms do not need to consider scheduling of accepted jobs. In stark contrast, for the second type of elastic jobs, a well-designed market mechanism needs to pay close attention not only to whether to accept bids, but also when to schedule their execution according to deadline information. For example, consider a user who needs to periodically train machine learning models. He submits a training job that can finish in 2 hours if assigned to his specified GPU server. However, the user's model does not need real-time updates. He can accept getting the training results within 2 days. This leaves ample time room for job scheduling. A well-designed scheduling algorithm should wisely leverage such time tolerance to maximize resource utilization and social efficiency - for example, schedule the job within the time window tolerable to the task, choosing relatively low demand periods.

Zhou et al. \cite{Zhou2017AnEC} builds upon online auctions that explicitly handle jobs with firm deadlines, and further allows cloud users to express soft deadlines described by preferred job completion times and penalty functions that encode the penalty amounts associated with different degrees of deadline violations. Compared to simple market mechanisms like fixed pricing, carefully designed auctions can automatically discover prices and timely adjust prices according to fluctuations in supply and demand, allocating cloud resources to the most valuable jobs to maximize the overall "happiness" of the users.

This paper builds upon this work by further expressing the preferences of both the service provider and users, and introduces the concept of soft acceptance. Compared to fixed online auction-based market mechanisms, soft acceptance stimulates users' true valuations more, and can timely adjust prices according to supply-demand relationships when resources are tight, allocating cloud resources to the most valuable jobs to maximize overall welfare.

We aim at the following goals in our cloud auction design. First, we require the cloud auction to be computationally efficient and execute in polynomial time. Second, the auction should be truthful, so that bidding true valuations of jobs is the dominant strategy for cloud users. Third, the auction should maximize the social welfare of everyone (including both the cloud provider and cloud users) in the system. Such a cloud auction design faces many challenges. First, truthfulness is a quite strong property that can only be achieved by a carefully prepared pair of virtual machine allocation and payment algorithms working together. Moreover, even if we can assume cloud users are altruistic and bid truthfully for free, the winner determination problem for social welfare maximization is an integer linear program (ILP) that is difficult to solve and NP-hard. The work in \cite{Zhou2017AnEC} has completed relevant proofs and online algorithm designs that can decide immediately upon arrival of each bid. By making jobs divisible, we propose an additional condition that the service provider can decide to accept first, and then pause or terminate the service when resources are not enough, called soft acceptance. This allows satisfying all user requests when resources are abundant, and satisfying high-value tasks when resources are insufficient, achieving social welfare maximization.

We first consider a basic setting where resources in the cloud are free within known capacity limits, and soft deadlines can be represented by enumerating some hard deadline options and their corresponding bid prices. We first present the natural ILP formulation of the social welfare maximization problem. Although the ILP is of polynomial size, it involves both conventional constraints (capacity limits) and unconventional constraints (job deadlines). With the soft acceptance condition, we consider this by pricing each user's per-time-slot value based on their different bid prices for different soft deadlines.

For this, we design an online algorithm that, under the added condition of soft acceptance, allocates tasks in the current time slot to all users if resource constraints are satisfied, and otherwise evaluates the per-time-slot value based on soft deadlines, pushing the lowest value task to the next time slot, guaranteeing the maximum value property of all current tasks, until all tasks finish or are paused-rejected. We theoretically analyze its competitive ratio and prove resource upper bounds on when the online algorithm can achieve the optimal.

We continue to promote cloud auction design by solving practical problems. First, we model the resource allocation costs in data centers using a convex cost function that characterizes server costs through dynamic voltage and frequency scaling \cite{wiki2023dynamic}. Second, we consider a general form of soft deadlines specified by (i) a preferred deadline and (ii) a non-decreasing penalty function for deadline violations. The new social welfare maximization problem is an integer convex program. (iii) We design a soft accept protocol where tasks are terminated, intermediate results are returned, and bids are refunded when tasks cannot proceed. Our published pricing auction framework builds upon the first two to consider soft accepts as a general scenario setup.

In the following sections, we first review related prior work in Sec. \ref{related_word}. Sec. \ref{system_modeling} then describes our system model and problem formulation. Our proposed framework for cloud auctions is presented in Sec. \ref{framework}, with accompanying analysis in Sec. \ref{analysis}. Simulation experiments evaluating the framework are detailed in Sec. \ref{experience}. Finally, Sec. \ref{conclusion} concludes the paper and summarizes key contributions.

\section{RELATED WORKS}\label{related_word}

The design of market mechanisms for cloud computing, especially auction mechanisms for cloud resource trading, has attracted great interest from the research community and given rise to a plethora of virtual machine auctions in the past few years\cite{Zhang2014DynamicRP,Shi2014AnOA,Zhang2015OnlineAI,Zhou2017AnEC,Wang2012WhenCM,Tan2020OnlineCA,Chen2021EdgeDRAO}.

The earliest virtual machine auctions were simple because they were one-round auctions and assumed clouds provide a single type of virtual machine, or virtual machine configurations that linearly scale \cite{Wang2012WhenCM}. They also assumed scenarios of static VM configuration where the number and types of VMs to sell are predetermined before the auction starts \cite{Zaman2011CombinatorialAD}.

Dynamic VM configuration, where the cloud provider decides which and how many VMs to assemble based on demand learned from user bids during the auction, has been studied in the past two years \cite{Zhang2014DynamicRP}–\cite{Zhang2015OnlineAI}. Zhang et al. designed a randomized auction for dynamic resource provisioning in cloud computing based on convex decomposition techniques, which is truthful and guarantees a small approximation ratio on social welfare \cite{Zhang2014DynamicRP}. Shi et al. further studied dynamic resource provisioning with budget-constrained cloud users, and designed online auctions where decisions are coupled in time due to fixed user budgets \cite{Shi2014AnOA}.

Online cloud auctions appear later than their one-round counterparts. Zhang et al. is among the first to study online cloud auction design, but they assume all VMs are of a uniform type \cite{Zhang2013AFF}. 
The work of Shi et al. \cite{Shi2014AnOA} designs online auctions, but does not consider the temporal correlation in decision making due to jobs spanning multiple time slots. 

A recent work of Tan et al. \cite{Tan2020OnlineCA} study the online combinatorial auction resource allocation problem with supply costs and capacity constraints. In the model, the supplier charges customers buying a bundle of resources, and the supplier's supply cost increases relative to the total resources allocated. We focus on maximizing social welfare, adopt the competitive analysis framework, and provide an optimal online mechanism through posted pricing.

Chen et al. \cite{Chen2021EdgeDRAO} designs procurement auctions to address operator-tenant incentive issues in distributed cloud demand response, enabling operators to optimize workload allocation and resource scheduling, but may incur switching costs when opening or closing clouds for grid stability.

Fan et al. \cite{Fan2017DeadlineAwareTS} proposed a deadline-aware task allocation mechanism, and transformed the task scheduling problem into a multi-dimensional 0-1 knapsack problem. They adopted an efficient task allocation algorithm based on ant colony optimization, which improves the total profit of the system while meeting task deadlines and resource constraints.

The work in \cite{Zhou2017AnEC} considers cloud auctions with job elasticity and job execution deadlines, and proposes compact index encoding optimization techniques using the primal-dual method \cite{Buchbinder2009TheDO} to effectively handle the new job deadline constraints in cloud social welfare maximization. It designs a series of online algorithms that operate on the dual problem. Designing the cost function for the dual problem also requires some trial and error. Even with abundant resources, poor cost function design can lead to loss of users.

Compared to existing cloud auction literature, our work designs a locally-aware solution approach for cloud auctions with elastic jobs and multiple soft deadlines, effectively handling the original problem without needing to design pricing functions induced by the dual problem. Building upon an offline solvable basis, we further enable online applicability by introducing soft acceptance conditions and adopting a online model.

\section{Systems and Modeling}\label{system_modeling}

In system modeling work, we consider a cloud data center hosting $M$ types of resource pools, including CPU, RAM, and disk storage that can be dynamically assembled into different types of virtual machines (VMs). Let $[\text{X}]$ denote the integer set ${1,2,\ldots, \text{X}}$. There are a total of $c_m$ resource units of type $m$ in the cloud. The cloud service provider serves as an auctioneer, renting out VMs to cloud users through auctions. Users' bids will arrive randomly at time-slots $1,2, \ldots, S$ over a long timespan.As shown in the figure \ref{schematic}, this is a problem of cloud computing resource auction and task scheduling. Note that multiple bids can arrive simultaneously and will be randomly ordered. There is a set of users $U$ participating in the auction, where each user $i$ requests multiple types of VMs and specifies in their bid: (i) $r_i^m$, the total amount of type-$m$ resources, and (ii) $s_i$, the specified time slots length for the VM to complete the job. The job execution is discrete and does not need to be continuous. User $i$'s job can execute at any time slots as long as the total execution time is within the deadline and satisfies $t_i$.

\begin{figure}
	\centering
	\includegraphics[width=0.8\linewidth]{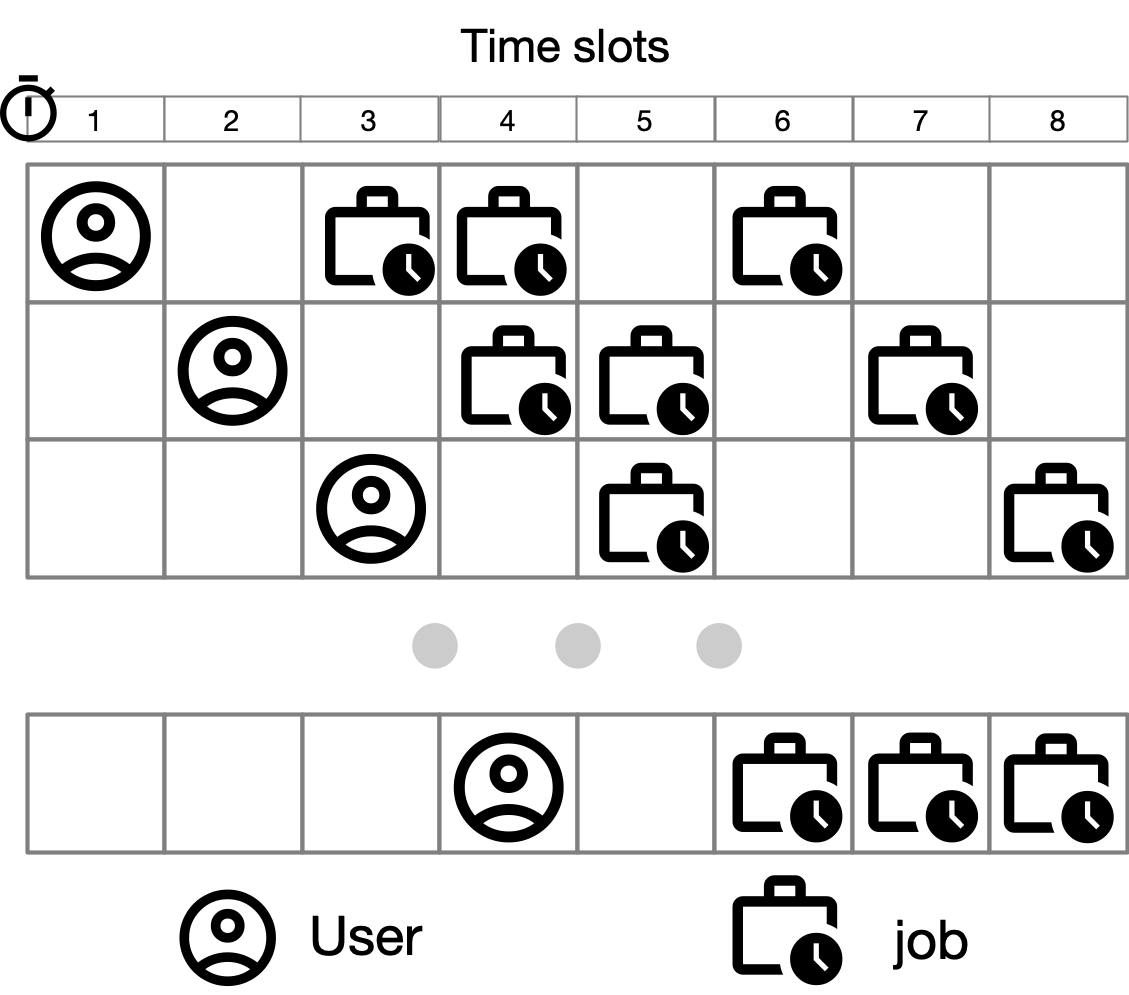}
	\caption{Schematic diagram of cloud computing resource task schedulingn}
	\label{schematic}
\end{figure}

\subsection{Bids with Multiple Deadlines}
First we consider a basic scenario where each user submits $B$ optional bids to express separate deadline options. User $i$'s bids consist of the required resource type list $r_i^m, \forall m$; the requested slot amount $s_i$; and the final deadline $e_{ij}, \forall j$, each with a corresponding bid price $p_{ij}$. We use $A_i$ to denote user $i$'s auction language submitted at time $a_i$:
$$
A_i=\left\{a_i,\left\{r_i^m\right\}_{m \in[M]}, s_i,\left\{e_{i j}, p_{i j}\right\}_{j \in B}\right\} .
$$
 
We adopt the XOR bidding rule, assuming each user can win at most one bid out of the $B$ optional bids \cite{Zhang2014DynamicRP}. Upon arrival of each bid, the cloud provider will consider the urgency of the user's demand for task scheduling, and try to complete the task according to the user's wishes as much as possible. After completing the task, it will inform the user of the completion status and actual winning price. If user $i$'s $j_{th}$ bid wins, then the binary $x_{ij}$ equals 1, otherwise 0. Let another binary variable $y_i(t)$ encode user $i$'s schedule: if user $i$'s job is scheduled to run at time $t$, then $y_i(t)=1$, otherwise 0.

In practice, users are assumed to be selfish, naturally aiming to maximize their own utility; they may misreport their true valuations in order to gain higher utility. In contrast, the cloud provider pursues the highest possible social welfare, making everyone in the cloud system "happy". Therefore, for the cloud provider, eliciting the true bids is very important. Thus, the cloud provider's consideration of the urgency of users' demands and trying to complete tasks according to their wishes can effectively motivate users to reveal their true valuations to some extent. This helps align user and system objectives for maximizing social welfare.

\begin{table}[!t]
	\caption{Notations}
	\centering
	\begin{tabular}{c l}
		\hline
		\hline
		Inputs & Descriptions\\
		\hline
		[X] & integer set {1, 2, ... ,X}\\
		U & set of Users\\
		B & set of bids per user \\
		S & set of slots\\
		$a_i$ & user $i$'s arrival time \\
		$r_i^m$ & demand of type-$m$ resource by user $i$ \\
		$s_i$ & \# slots requested by user $i$ \\
		$e_{i j}\left(e_i\right)$ & deadline of user $i$'s $j$th (user $i$'s) bid \\
		$p_{i j}\left(b_i\right)$ & bidding price of user $i$'s $j$th (user $i$'s) bid\\
		$c_m$ & capacity of type-$m$ resource  \\
		\hline
		Decisions & Descriptions \\
		\hline
		$x_{i j}\left(x_i\right)$ & Whether user $i$'s $j$th bid wins (1) or not (0) \\
		$y_{it}$ & whether to allocate user $i$ s job in slot $\mathrm{t}$ \\

		\hline
		\hline
	\end{tabular}
\end{table}

\subsection{Problem Formulation}

The social welfare maximization problem with alternative deadlines can be formulated into the following ILP:

\begin{subequations}
	\begin{align}
		\operatorname{maximize} & \sum_{i \in U} \sum_{j \in B} p_{i j} x_{i j} \label{main_eq} \\
		\text { subject to: } \quad & \sum_{j \in B} x_{i j} \leq 1, \forall i \in U, \label{limit_job}\\
		& y_{it} t \leq \sum_{j \in B } e_{i j} x_{i j}, \forall t \in S , \forall i \in U : a_i \leq t,\label{limit_deadline} \\
		& \sum_{j \in B} s_i x_{i j} \leq \sum_{t \in T: a_i \leq t} y_{it}, \forall i \in U, \label{limit_win_bid}\\
		& \sum_{i \in U: a_i \leq t} r_i^m y_i(t) \leq c_m, \forall m \in M, \forall t \in S, \label{limit_resource}\\
		& x_{i j}, y_{it} \in\{0,1\}, \forall i \in U, \forall t \in S, \forall j \in J .
	\end{align}
\end{subequations}

We can see that this is a multi-dimensional 0-1 knapsack problem, and is NP-hard. The multiple deadlines are modelled with the XOR bidding rule by (\ref{limit_job}), constraint (\ref{limit_deadline}) ensures that a job is scheduled to run between its arrival time and deadline. Constraint (\ref{limit_win_bid}) guarantees that the number of allocated slots is sufficient for serving a successful bid, and the capacity limit of each type of resource is expressed in constraint (\ref{limit_resource}).

Even in an offline environment, without constraints (\ref{limit_deadline}) and (\ref{limit_win_bid}), the ILP (\ref{main_eq}) is still an NP-hard combinatorial optimization problem, equivalent to the classical knapsack problem. When we deal with job deadlines and make online decisions, the challenges are further escalated. To tackle these challenges, we first design a framework to handle modeling constraints of multiple soft deadlines, more specifically, we transform the original ILP problem into a locally-aware cumulative maximization problem.

\section{Framework Introduction}\label{framework}

For the aforementioned ILP optimization problem, this is an NP-hard problem. Therefore, we will redesign it and propose a computational framework, which mainly consists of three modules: task allocation, evaluation, and action. 1).Allocation phase: Allocate the current unfinished tasks to all users at the current time. 2).Evaluation: After allocation, evaluate whether the allocation in each time slot meets the requirements based on some metric (e.g. resource usage). 3).Action: Obtain the evaluation results. For tasks that do not meet the time slot requirements, find the next available time slot for each user. If found, adjust the allocation. If not, cancel the task until the end of the time horizon.
As shown in the figure, it demonstrates one kind of dynamic task allocation method.

\begin{figure}[!h]
	\centering
	\includegraphics[width=3.6in]{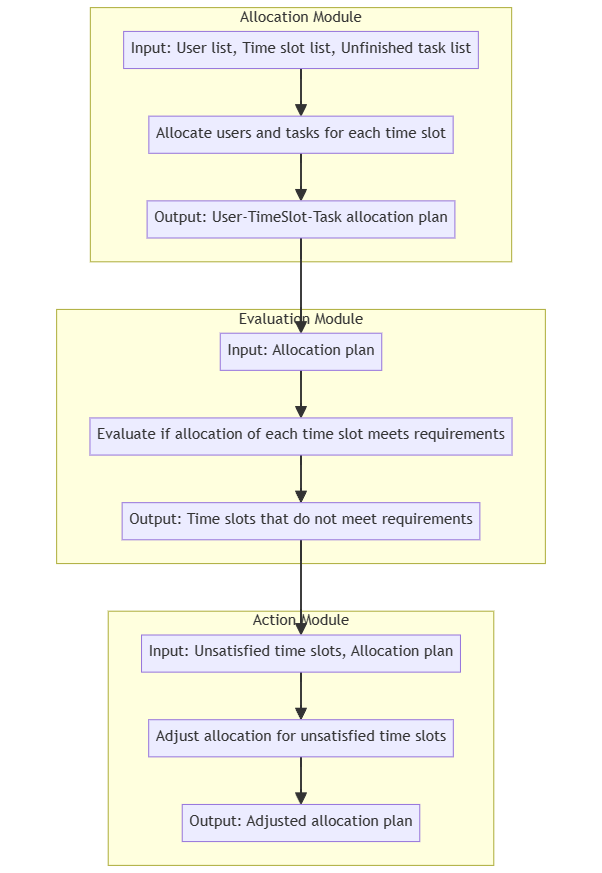}
	\caption{Auction Task Allocation Framework}
	\label{framwork}
\end{figure}

\begin{algorithm}
	\caption{Auction Time Scheduling Framework}\label{alg:alg1}
	\begin{algorithmic}
		
		\FOR{$t=1$ to $S$}
		
		\STATE {\bfseries Allocation Module:}
		
		\STATE Input: User list, Time slot list, Unfinished task list
		
		\STATE Allocate users and tasks for each time slot
		
		\STATE Output: User-TimeSlot-Task assignment
		
		\STATE {\bfseries Evaluation Module:}
		
		\STATE Input: Assignment
		
		\STATE Evaluate if allocation in each time slot meets requirements
		
		\STATE Output: Time slots that do not meet requirements
		
		\STATE {\bfseries Action Module:}
		
		\STATE Input: Unmet time slots, Assignment
		
		\STATE Adjust allocation for unmet time slots
		
		\STATE Output: Adjusted assignment
		
		\ENDFOR
		
	\end{algorithmic}
\end{algorithm}

In order to optimize the time scheduling of the auction platform, we designed an auction time scheduling framework based on cyclic iterations. This framework contains three key modules: Allocation Module, Evaluation Module, and Action Module. In each time cycle, the Allocation Module generates an initial mapping of time slots to users and tasks based on the user list, time slot list, and unfinished task list. Then the Evaluation Module evaluates the allocation scheme for each time slot based on some optimization objectives and constraints. For those time slots that do not meet the requirements, the Action Module tries to improve them by adjusting their allocation schemes. After this round of iteration, an adjusted new allocation scheme is generated as the new input for the next cycle. Through such cyclic scheduling iterations, each module works collaboratively to gradually optimize and converge to a relatively optimal allocation of time slots to users and tasks. Compared to one-shot static allocation, this framework achieves dynamic optimization by introducing closed-loop evaluation and adjustment. Meanwhile, compared to separate evaluation or adjustment, it comprehensively considers the interactions between allocation, evaluation and adjustment. This closed-loop design based on control theory provides an effective and general solution for auction time scheduling.

\begin{definition}
	\textbf{User's Unit Time Value} The user's unit time value (Unit Time Value) refers to the user's emphasis on a task i at time t. At current time t, within the j-th soft deadline ($e_{ij}$), it can be represented by the ratio of the user's j-th bid ($p_{ij}$) to the time difference between the start time ($a_i$) and the j-th soft deadline ($e_{ij} - a_i$), i.e.:
	\begin{align}
		u_{it} = \frac{p_{ij}}{e_{ij} - a_i}
	\end{align}
\end{definition}

This definition reflects the change in the price a user is willing to pay to complete a task as time goes by. The shorter the time difference, the higher the user's unit time value, indicating the user's increasing urgency to complete the task. The calculation method is the user's bid divided by the corresponding time interval. This concept can help understand and predict the user's time sensitivity and task urgency.

\begin{definition}
	\textbf{Unit Time Resource Value Density} The unit time resource value density is an indicator that reflects the user's time sensitivity and resource sensitivity in obtaining the required resources to complete a task or output.
	
	\begin{align}
		v_{it}^m = \frac{u_{it}}{r_i^m}
	\end{align}
	
\end{definition}

Where $v_{it}^m$ denotes unit time resource value density of user i for resource type m at time t.$u_{it}$ denotes user i's unit time value for the task at time t. $r_i^m$: Amount of resource type m required by user i to complete the task.

This definition reflects how users value obtaining a unit amount of resources required to complete a task or output over time. It takes into account the effects of both time factors and resource factors on user decisions. The higher the unit time resource value density, the higher the user's time sensitivity and resource sensitivity towards the task or output, and the higher the urgency to complete the task or obtain the output.
This indicator can be used to analyze the user's time elasticity and resource elasticity, assess the urgency of tasks in different periods, and provide a reference basis for relevant decisions such as resource scheduling and pricing.
The unit time resource value density simultaneously considers time factors and resource factors. A higher value means the user has higher urgency and willingness to pay to obtain resources to complete tasks or outputs, reflecting higher time and resource sensitivity. This metric is useful for understanding user behavior and making scheduling and pricing decisions accordingly.

\subsection{Offline Design}

Based on the framework in Algorithm \ref{alg:alg1} and the definitions, design an offline algorithm where when a task arrives at time $a_i$, allocate the task based on its required length $s_i$ to every time slot from $a_i$ to $a_i + s_i$. Check if at current time $t_i$, the resource usage of all users exceeds the limit $c_m$. If exceeded, evaluate each resource and based on the defined per time unit value of the user and per time unit resource value density, evaluate the user with the lowest per time unit resource value density. Cancel the current allocation of the user with the lowest score and postpone it to the next available time slot. If no available time slot, cancel the task, until the task is finished. When the task is completed, give the final price. This completes the algorithm. The pseudocode of the algorithm is shown in Algorithm \ref{resource_allocate_alg}.

\begin{algorithm}
	\caption{Resource Allocation Adjustment Algorithm}\label{resource_allocate_alg}
	
	\textbf{Input}: Time slot list, Resource capacity, Original user-timeslot-resource allocation
	
	\textbf{Output}: Adjusted user-timeslot-resource allocation
	
	\begin{algorithmic}[1]
		
		\FOR{t = 1 to S for each time slot}
		
		\STATE Allocate unfinished tasks to current time slot
		
		\FOR{m = 1 to M for each resource type}
		
		\STATE Calculate current resource usage $R_i^m$
		
		\IF{$R_i^m > r_i^m$}
		
		\STATE Evaluate current users, find all users exceeding capacity $u_i$
		
		\FOR{each exceeding user $u_i$}
		
		\STATE Get this user's next available time slots $s_r$ after current
		
		\IF{available time slots exist}
		
		\STATE Postpone this user's task to the first slot in $s_r$
		
		\ELSE
		
		\STATE Reject user, cancel this user's allocation in all previous slots
		
		\ENDIF
		
		\STATE Recalculate resource usage $R_i^m$
		
		\ENDFOR
		
		\ENDIF
		
		\ENDFOR
		
		\ENDFOR
		
		\RETURN new schedule table
		
	\end{algorithmic}
	
\end{algorithm}

The proposed resource allocation adjustment algorithm dynamically adjusts user-timeslot-resource allocations to resolve resource overloading issues. It iterates through each time slot (line 1-2), allocating unfinished tasks to the current slot (line 2). For each resource type, it calculates current usage (line 5) and compares it to capacity (line 6).
If usage exceeds capacity, overloaded users are identified (line 7-9). For each overloaded user, the algorithm checks for available future time slots (line 10-11). If available, the user's task is postponed to the earliest future slot (line 12-13). Otherwise, the user is rejected and their allocation canceled (line 14-16).
Resource usage is then recalculated (line 18) before returning the adjusted allocation table (line 20).
The key steps are identifying overloaded users (line 7-9) and rescheduling or rejecting them (line 10-17) to ensure capacities are not exceeded. By limiting per-user resource usage and postponing or rejecting requests when resources are insufficient, the algorithm dynamically adjusts allocations to resolve overloading.

The time complexity analysis of the algorithm \ref{resource_allocate_alg}. This algorithm has two main loops: the outer loop iterates through each time slot, where the number of time slots is S, and the inner loop iterates through each resource type, where the number of resource types is M. In the inner loop, the following operations are performed: calculate resource usage, time complexity O(1); if resource usage exceeds capacity, iterate through each user ui exceeding capacity, where the number of such users is I, get the user's available time slots, time complexity O(1), postpone or reject the user, time complexity O(1); recalculate resource usage, time complexity O(1). So the time complexity of the inner loop is O(I). Therefore, the overall time complexity of the algorithm is: O(S) x O(M) x O(I) = O(SMI). That is, the time complexity of this algorithm is polynomial, and related to the number of time slots S, resource types M, and users exceeding capacity I. It can be denoted as O(SMI), and since the total duration is much greater than the number of users, I is far smaller than S.

\subsection{Evaluation Method}

There are 2 key points to note in the algorithm: 1).The user's per unit time value does not have to be limited to this method, as long as it reflects the user's urgency for the task. 2).For evaluating the algorithm, the evaluation algorithm's performance directly determines the advantage and disadvantage of the algorithm. In order to distinguish users with lower efficiency per unit time, make corresponding adjustments to their resources, so as to improve the overall resource utilization efficiency.

The resource value density per unit time is used as an indicator to compare the resource utilization efficiency between users. Therefore, algorithm \ref{alg_fun1} is designed. Specifically, this algorithm first initializes the recorded minimum resource value per unit time $v_{min}$ and the corresponding user ID $u_i$. Then it iterates through all users $i$ and resource types $m$, calculating the resource value per unit time $v_{it}^m$ for user $i$ using resource $m$ at time $t$, which is $\frac{u_{it}}{r_{im}}$. Here $u_{it}$ is the per unit time value of resource $m$, and $r_{im}$ is the amount of resource $m$ occupied by user $i$. If the calculated $v_{it}^m$ is lower than the recorded $v_{min}$, then $v_{min}$ and $u_i$ are updated. After the iteration ends, the recorded $u_i$ is the user with the lowest resource utilization efficiency. By comparing the resource value density per unit time, this algorithm intuitively reflects the resource utilization efficiency of different users on the system resources. Compared with the absolute resource occupancy, it can more effectively identify "occupying more resources but low value density" inefficient users.

\begin{algorithm}
	
	\caption{Time-Based Resource Utilization Evaluation Method (T-RUEM)} \label{alg_fun1}
	
	\textbf{Input}: Resource usage $r_i^m$ and resource value density per unit time $u_{it}$ at time $t$
	
	\textbf{Output}: User id $u_i$
	
	\begin{algorithmic}[1]
		
		\STATE Initialize $u_i = -1, v_{min} = +\infty$
		
		\FOR{$i \in I$ where I is the subset of users exceeding resources at current time}
		
		\FOR{$k \in M$}
		
		\STATE Calculate resource value density per unit time for user $i$ resource $k$ at time $t$: $v_{it}^m \gets \frac{u_{it}}{r_i^m}$
		
		\IF{$v_{it}^m < v_{min}$}
		
		\STATE Record minimum resource value density per unit time: $v_{min} \gets v_{it}^m$
		
		\STATE Record current user id: $u_i \gets i$
		
		\ENDIF
		
		\ENDFOR
		
		\ENDFOR
		
		\RETURN $u_i$
		
	\end{algorithmic}
	
\end{algorithm}

This algorithm contains two nested for loops, the outer loop iterates through all users exceeding resources I, and the inner loop iterates through all resource types M. Within the two for loops, there are mainly some simple arithmetic calculations, whose time complexity is O(1). Therefore, the overall time complexity of the algorithm is: O(I) x O(M) x O(1) = O(IM).
That is to say, when the number of users exceeding resources I or the number of resource types M is large, the running time of this algorithm will grow linearly with the increase of I or M.
But compared to directly calculating and comparing the absolute resource amounts of all users, this algorithm calculates the "resource value density per unit time", which actually avoids the calculation and comparison of absolute resource amounts. So it can be considered a more flexible design, thus reducing the amount of computation.And usually in practical situations, I and M are relatively small. 
Overall, the time complexity of this algorithm is polynomial level in the worst case, which is acceptable.

\textbf{Weighted Average Algorithm} Since the solving method of T-RUEM only considers the problem of one resource each time, without considering the occupancy of other resources, according to the proposed multi-resource allocation problem, the calculation method is redesigned and a weighted average algorithm based on unit resource value is proposed.
\begin{subequations}
	\begin{align}
		v'_{it} &= \frac{\sum_{m=1}^{M} v_{it}^m * r_i^m}{\sum_{m=1}^{M} r_i^m} \\
		&=   \frac{\sum_{m=1}^{M} \frac{u_{it}}{r_i^m} * r_i^m}{\sum_{m=1}^{M} r_i^m}	\\
		&= \frac{ u_{it}}{\sum_{m=1}^{M} r_i^m}
	\end{align}
\end{subequations}

Where $v'{it}$ - the weighted average value of all resources for user i at time t;
$v_{it}^m$ - the value of resource m to user i at time t, defined as the utility $u_{it}$ of user i divided by the amount of demand $r_i^m$ of resource m for user i. This can reflect the relationship between each user's utility and demand for resources.

\begin{algorithm}
	\caption{Time-Based Resource Weighted Average Evaluation Method (T-RWAEM)}
	
	\textbf{Input}: Resource usage $r_i^m$ and resource value density per unit time $u_{it}$ at time $t$
	
	\textbf{Output}: User id $u_i$
	
	\begin{algorithmic}[1]
		
		\STATE Initialize $u_i = -1, v'_{min} = +\infty$
		
		\FOR{$i \in I$ where I is the subset of users exceeding resources at current time}
		
		\STATE Initialize total weights $W \gets 0$
		
		\FOR{$m \in M$}
		
		\STATE $W \gets W + r_{i}^m$
		
		\ENDFOR
		
		\STATE Calculate weighted average value $v'_{it} \gets \frac{u_{it}}{W}$
		
		\IF{$v'_{it} < v'_{min}$}
		
		\STATE Record current minimum weighted average resource value density per unit time: $v'_{min} \gets v'_{it}$
		
		\STATE Record current user id: $u_i \gets i$
		
		\ENDIF
		
		\ENDFOR
		
		\RETURN $u_i$
		
	\end{algorithmic}
	
\end{algorithm}

Here is an analysis of the time complexity of the T-RWAEM algorithm in an academic writing style: The initialization operations have a constant time complexity of O(1). The outer loop iterates through the subset I of overloaded users, with a time complexity of O(|I|), where I is the number of overloaded users.
The inner loop iterates through the set M of resource types, with a time complexity of O(M), where M is the number of resource types.
Calculation of the total weights W also has a time complexity of O(M).
Computation of the weighted average value v' has a constant time complexity of O(1).
The if condition check and minimum value update operations have a time complexity of O(1).
Therefore, the overall time complexity of the algorithm can be expressed as:
$T(n) = O(1) + O(I)[O(M) + O(M) + O(1)] = O(IM)$
The algorithm has a time complexity of O(|I|*|M|), where |I| and |M| represent the number of overloaded users and resource types respectively.

The key operations are the two nested loops, so the time complexity depends on the scale of the overloaded user subset and resource types. We can consider this algorithm to have a relatively low time complexity.

Further analyzing the relationship between algorithm time complexity and I and M: When the number of users I is very large and the number of resources M owned by a single user is relatively small, the algorithm time complexity tends to O(I); When the number of resources M owned by a single user is very large and the number of users I is relatively small, the algorithm time complexity tends to O(M).
In summary, the algorithm's time complexity is linear with the number of users I and the number of resources M per user. It belongs to linear time complexity. Compared with polynomial time complexity, this linear time complexity algorithm has higher time efficiency.
The algorithm fully considers the user's resource usage amount, and introduces the resource amount as weights, making the identification results more convincing. The design of calculating the weighted average resource value density is reasonable, avoiding the problem of only relying on resource value density per unit time to identify users.
Also, its logic is simple and intuitive, easy to implement.

\subsection{Online Design}

The offline algorithm design has been completed in the previous sections. In order to enable online decision making, a soft acceptance protocol will be added, which is a delayed payment model. Users will first be accepted in the auction, and when the service cannot be satisfied, the user will be rejected. If all the user's tasks are finally completed, the settlement will be based on the final completion time.

Soft Accept is an online auction algorithm mechanism. Its core idea is to flexibly accept or reject user requests based on the current system resource status. Specifically, the soft acceptance algorithm allows temporarily accepting more service requests, even if resources ultimately cannot be fully satisfied, it can still respond to users first and reject the service when resources are insufficient. Compared with traditional offline batch allocation auctions, the soft acceptance algorithm has the following significant advantages:

\begin{itemize}
	
	\item First, soft acceptance can significantly improve resource utilization and avoid waste and idle resources. The algorithm can accept more requests until resources are exhausted, thereby maximizing the use of system resources.
	
	\item Second, soft acceptance can greatly reduce the user's waiting response time to meet the user's real-time service needs. Compared with offline batch processing, soft acceptance achieves "real-time" service.
	
	\item Third, the soft acceptance mechanism can flexibly handle dynamically changing user requests, exhibiting better scalability and elasticity, which is very meaningful for many online services.
	
	\item Fourth, soft acceptance implements pay-after-service, lowering the barrier for users, who can get the service first and then pay after the request is finally accepted.
	
	\item Fifth, the service provider can ensure its revenue by setting payment prices to guarantee profits in case of insufficient supply.
	
\end{itemize}

Finally, soft acceptance simplifies business processes without requiring users to prepay. Its core algorithm is basically consistent with the offline algorithm, but it requires users to comply with the soft acceptance protocol. The soft acceptance online auction mechanism combines the real-time, elasticity and pay-after-service features of auctions, and compared with offline auctions, it has advantages such as high utilization, short waiting time, handling dynamic requests, and enabling pay-after-service, making online auction algorithms more efficient and practical.

\section{Theoretical Analysis}\label{analysis}

\begin{definition}
	Let $S$ be the set of all feasible solutions, $P$ be the objective function of the resource allocation problem.
	Let $\Delta$ be the maximum number of moves of the adjustment algorithm A, Let A be the resource allocation adjustment algorithm in Algorithm \ref{resource_allocate_alg}. Here $\Delta$ can take non-negative integers.
	Let $s^*$ be the optimal solution, $s$ be any feasible solution.
	Define the maximum number of moves to the optimal solution as $\Delta^*$, when $s^*$ is reachable, $\Delta \leq \Delta^*$.
\end{definition}

\begin{lemma}
	
	For the resource allocation problem, there exists an optimal solution $s^* \in S$, for any feasible solution $s \in S$, $P(s^*) \geq P(s)$.
	
\end{lemma}

\begin{proof}
	
	Step 1: Construct a feasible resource allocation solution algorithm A, for each time slot $t$, for resource type $k$ exceeding capacity:
	
	Let excess user $u$ postpone allocation to $u$'s next earliest available time slot $t'$.
	
	Step 2: Prove algorithm A minimizes the number of moves $\Delta$.
	
	By the construction of algorithm A, within the search space of solutions, the number of moves of this solution is minimized, with upper bound $\Delta$.
	
	Step 3: Let $s'$ be the feasible solution generated by algorithm A. Because $s'$ has the minimum number of moves, $P(s') \geq P(s), \forall s \in S$.
	
	Step 4: From Step 1 and 3 we have: $P(s^*) \geq P(s'), \forall s' \in S$.
	
	Step 5: When $\Delta \leq \Delta^*$, $\Delta$ represents the number of moves of algorithm A is minimized, falling within the solution space with minimum number of moves.
	
	According to algorithm A, at this time $s'$ has the minimum number of moves, so $s'$ is the optimal solution.
	
	In summary, the lemma is proved.
	
\end{proof}

\begin{proof}
	Now we prove: the resource allocation problem satisfies the optimal substructure property.
	
	Let the optimal solution of the resource allocation problem be $S^*$, which contains m subproblems $P_1,P_2,...,P_m$.
	
	For any subproblem $P_i$, let its optimal solution be $S_i^*$. Construct a new solution $S'$:
	
	$$S' = S_1^* \cup S_2^* \cup ... \cup S_m^*$$
	
	That is, $S'$ contains the optimal solutions of all subproblems.
	
	We prove: $S'$ is also the optimal solution to the original problem.
	
	Assume the opposite, that $S'$ is not the optimal solution to the original problem. That is, there exists another solution $S''$, such that:
	
	$$P(S'') > P(S')$$
	
	Where $P()$ is the objective function of the resource allocation problem.
	
	Decompose $S''$ into solutions for subproblems $P_1,P_2,...,P_m$, denoted as $S_1'',S_2'',...,S_m''$.
	
	Then according to the properties of the resource allocation problem, there must exist some $i$, such that:
	
	$$P(S_i'') > P(S_i^*)$$
	
	This contradicts that $S_i^*$ is the optimal solution to subproblem $P_i$.
	
	Therefore the original assumption does not hold. That is, $S'$ is the optimal solution to the original problem.
	
	This proves that the optimal solution to the resource allocation problem contains the optimal solutions to all its subproblems.
	
	Therefore, the resource allocation problem satisfies the optimal substructure property.
	
\end{proof}

\begin{lemma}
	
	In the feasible solution space $S$, if there exists a feasible solution $s'$ with the minimum number of moves, then $s'$ is the optimal solution.
	
\end{lemma}

\begin{proof}
	
	Assume in the feasible solution space $S$, there exists a feasible solution $s'$ with the minimum number of moves, denoted as $\Delta'_{\min}$.
	
	Assume $s'$ is not the optimal solution, that is, there exists another feasible solution $s^*$, with a more optimal objective function value, i.e. $P(s^*) > P(s')$.
	
	Since the resource allocation problem satisfies the optimal substructure property, transferring from any feasible solution to the optimal solution must require a certain number of moves. Let the number of moves from $s'$ to $s^*$ be $\Delta''$.
	
	Because $s'$ already has the minimum number of moves, we have: $\Delta'' > \Delta'_{\min}$.
	
	This contradicts the assumption that $s'$ has the minimum number of moves.
	
	By proof by contradiction, the original assumption "s' is not the optimal solution" does not hold.
	
	Therefore, if $s'$ has the minimum number of moves, it must be the optimal solution.
	
\end{proof}

\begin{definition}
	\label{define:competitive_ratio}
	
	Let the optimal solution be $s^*$, any feasible solution be $s'$. Define the competitive ratio of $s^*$ and $s'$ as:
	$\alpha(s',s^*) = \frac{P(s')}{P(s^*)}$
	Where, $\alpha(s',s^*)$ represents the competitive advantage of $s'$ relative to the optimal solution $s^*$.
	
\end{definition}

\begin{lemma}
	$\alpha(s',s^*)$ is monotonically decreasing with respect to the number of moves $\Delta$.
\end{lemma}

\begin{proof}
	
	(1) According to Lemma 1, in the feasible solution space, the solution with the minimum number of moves $\Delta$ is the optimal solution.
	
	(2) Let any two feasible solutions be $s'_1,s'_2$, with number of moves $\Delta_1,\Delta_2$ respectively, and $\Delta_1 < \Delta_2$.
	
	(3) From (1), we know $\Delta_1 < \Delta_2$, then $P(s'_1) \geq P(s'_2)$.
	
	(4) Substitute $s'_1,s'_2$ into the definition of $\alpha$:
	
	$\alpha(s'_1,s^*) = \frac{P(s'_1)}{P(s^*)} \geq \frac{P(s'_2)}{P(s^*)} = \alpha(s'_2,s^*)$
	
	(5) Therefore, when $\Delta_1 < \Delta_2$, we have $\alpha(s'_1,s^*) \geq \alpha(s'_2,s^*)$.
	
	(6) In summary, $\alpha(s',s^*)$ is monotonically decreasing with respect to $\Delta$.
	
\end{proof}

This section has proved the following three lemmas:

(1) In the feasible solution space, the solution with the minimum number of moves is the optimal solution. This proves that algorithm A can obtain the optimal solution by minimizing the number of moves.

(2) The solution with the minimum number of moves must be the optimal solution. This further strengthens the conclusion of Lemma 1.

(3) The defined competitive ratio decreases monotonically as the number of moves decreases. This proves that minimizing the number of moves can continuously improve the solution quality.

In summary, these three lemmas prove that algorithm A can produce the solution with the minimum number of moves, i.e. the optimal solution, by greedily minimizing the number of moves. The lemmas provide theoretical guarantees that algorithm A produces the optimal solution. This provides a simple and efficient optimization algorithm for the resource allocation problem.

\section{Performance Evaluation}\label{experience}

In order to evaluate the performance of the proposed algorithm, we designed comprehensive experiments under different experimental settings. On the basis of using Google Cluster Data\cite{google2023cluster}, we designed each job as a bidding behavior, added some simulated data, which contains information about each job, including start time. We assume each user's task is [1, 12] time slots, and each time slot is 5 minutes. Specifically, we generated a comprehensive dataset with the number of users, resource types. We also introduced baseline methods for comparison, including random algorithm, greedy algorithm.

In the comparative experiments of this paper, the random method randomly assigns tasks within the latest time required by the user after the user arrives, and then selects the corresponding bidding based on the completion time of the assigned tasks. The greedy algorithm finds the time node with the highest bid within the latest time required by the user after the user arrives, tries to meet the user's requirements as much as possible, and satisfies the user's requirements in a secondary manner if there are not enough resources. If the resources or time cannot meet the requirements, the task will be canceled. In the competitive bidding experiments of this paper, as mentioned in \ref{define:competitive_ratio}, it is the ratio of the current target welfare to the maximum welfare.

Experiment 1 aims to discuss the competitive ratio of the algorithm. It designs different resource allocation algorithms under different number of users, resource types and other conditions to compare the competitive ratio. The experiment compares the random method, greedy method, T-RUEM and T-RWAEM, as well as the optimal solution obtained by the solver gurobi. Since the task time slots and bidding prices may vary, in order to enhance persuasiveness, the experiment is performed 100 times to take the average value. Considering the complexity of this problem and the time limit that can be solved, the number of users is set between 10 and 50, the resource types are fixed to 2 types, including CPU and RAM, and the resource limit is normalized to 1. The experimental results are shown in Figure \ref{fig:exp1}.

\begin{figure}[htb]
	\centering
	\includegraphics[width=4.5in]{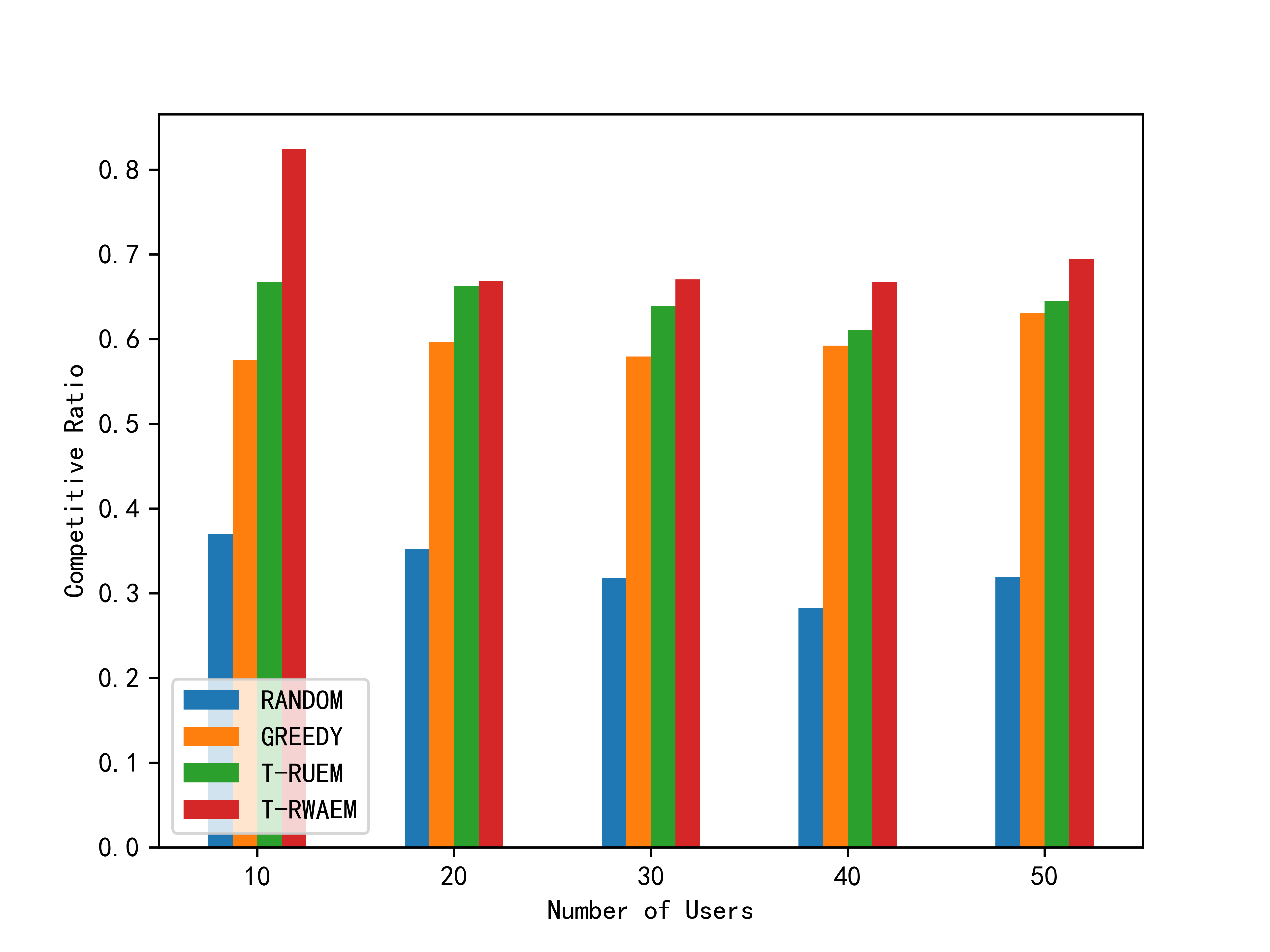}
	\caption{Competitive ratio of above algorithms with  different number of users}
	\label{fig:exp1}	
\end{figure}

As the number of users increases, the overall performance of the random method shows a downward trend, which is consistent with the inherent characteristics of the random method. The performance of the greedy method is relatively stable and overall better than the random method. The performance of both the T-RUEM method and the T-RWAEM method is better than the random method and the greedy method, and the T-RWAEM method is significantly better than the T-RUEM method. When the number of users increases from 10 to 50, they can maintain the lead and the average competitive ratio of the T-RWAEM method is 0.71, which shows that the T-RWAEM method can provide better performance in user groups of different scales. In summary, this group of experiments verifies the validity of the proposed T-RUEM method and T-RWAEM method. Compared with the traditional random method and greedy method, the proposed framework can achieve a higher approximation ratio.

Experiment 2 aims to discuss the comparison of social welfare of the algorithm under the condition of a large number of people. It designs different resource allocation algorithms under different number of users, resource types and other conditions to compare social welfare. The experiment compares the random method, greedy method, T-RUEM and T-RWAEM. The simulated settings are the same as experiment 1, except that the number of users is 100-200, the number of bids is 3, 6, 9, respectively, the resource types are still fixed at 2 types, including CPU and RAM, and the resource limit is normalized to 1. The results are shown in Figure \ref{fig:exp2-total}. Analyzing the experimental results as a whole, it is not difficult to find that T-RUEM and T-RWAEM still have higher social welfare than the baseline. In order to further analyze the impact of different bids, the results of the proposed algorithms are displayed respectively, as shown in Figures \ref{fig:exp2-T-RUEM} and \ref{fig:exp2-T-RWAEM}, which show the performance comparison of algorithms T-RUEM and T-RWAEM when the number of bids is 3, 6, 9, displaying the impact of multiple bidding prices on social welfare under different number of users. It shows that the two algorithms have similar adaptability to changes in the number of bids, which verifies the effectiveness of the proposed algorithms in large-scale scenarios and shows their adaptability to changes in the number of users and bids. This provides support for the scalability of the algorithms in practical applications.
\begin{figure}[!htb]
	\centering
	\includegraphics[width=4.5in]{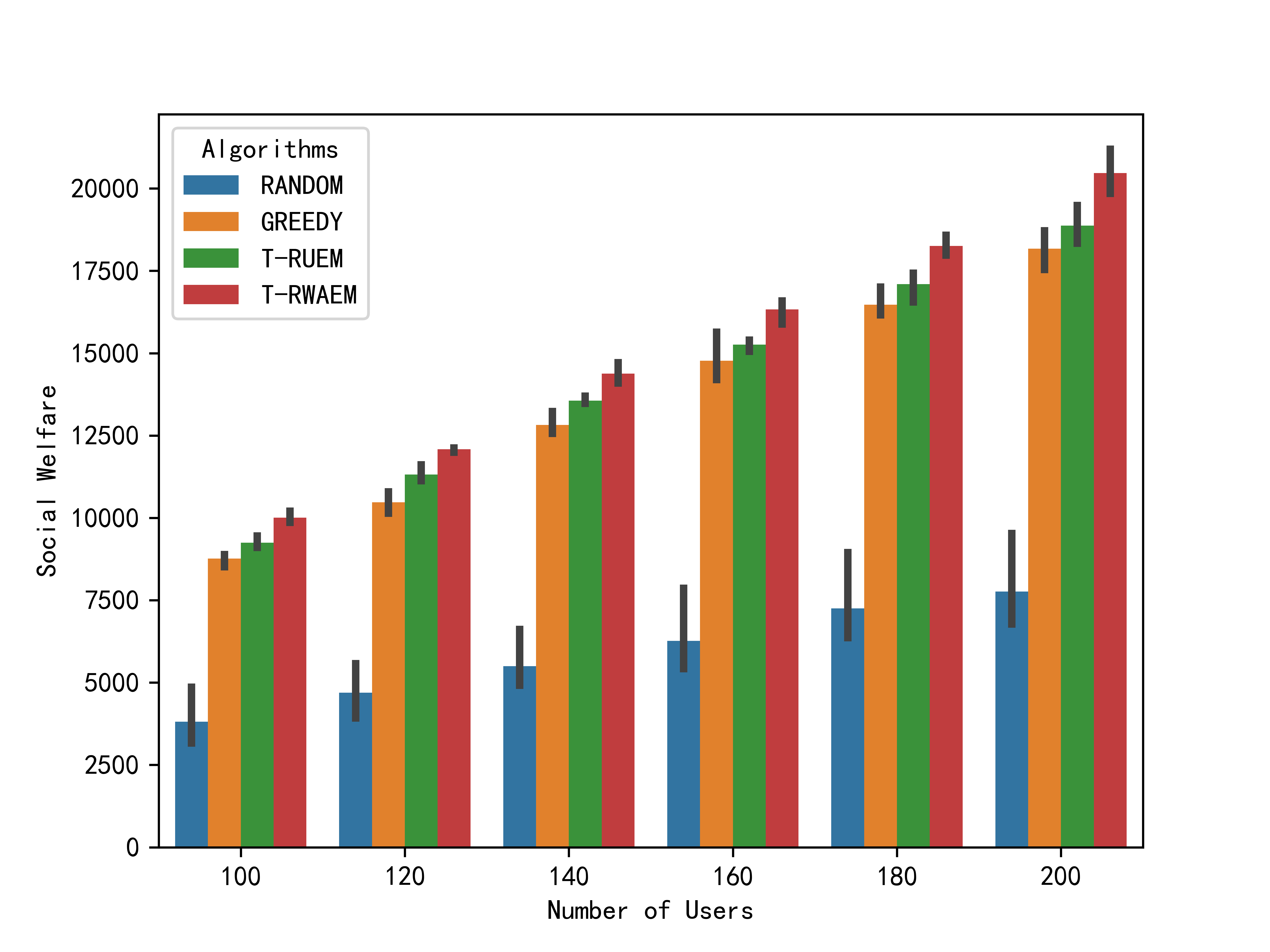}
	\caption{Comparison of social welfare when the number of different bids}
	\label{fig:exp2-total}
\end{figure}

\begin{figure}[!htb]
	\centering
	\includegraphics[width=4.5in]{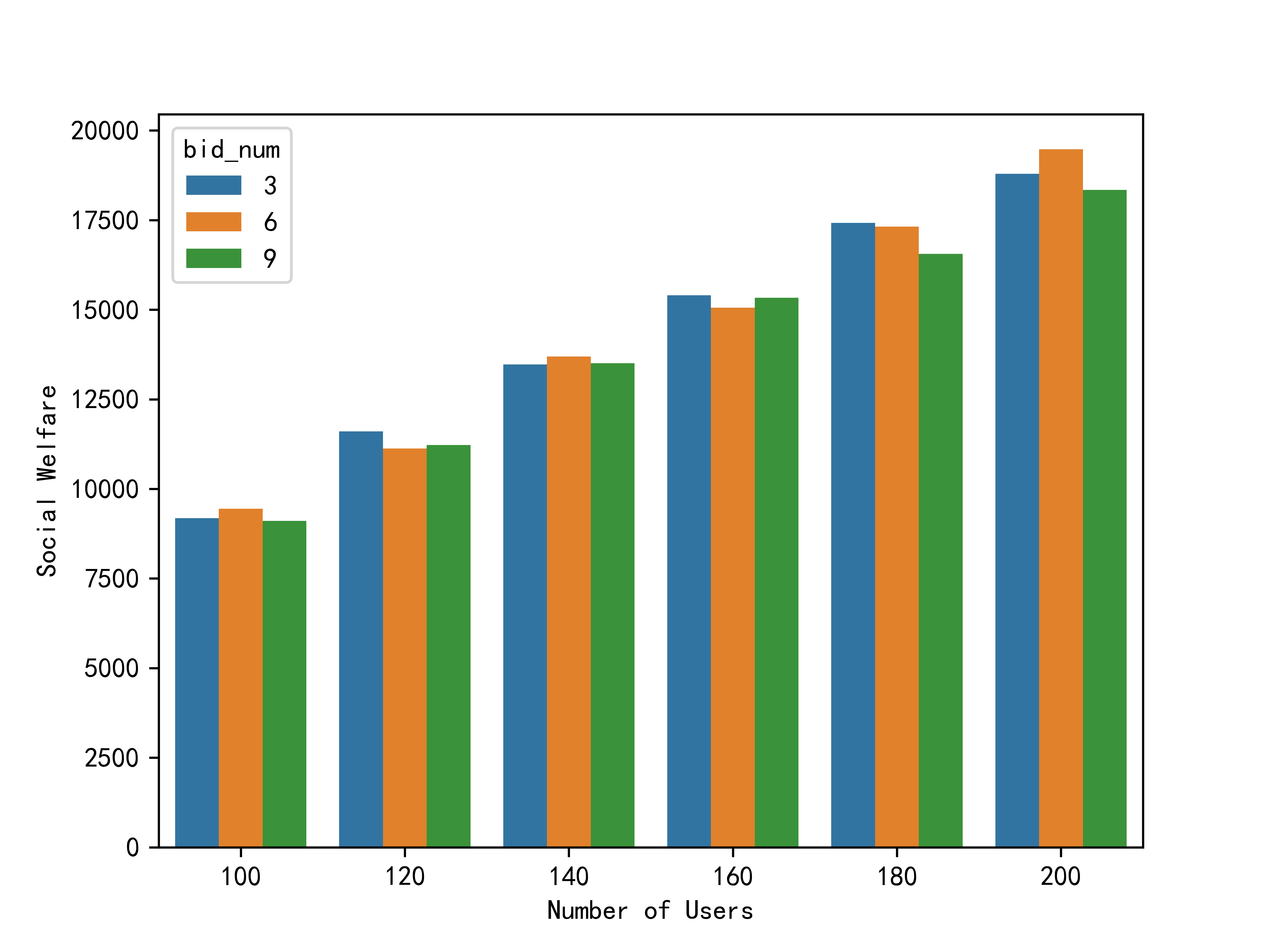}
	\caption{Comparison of social welfare of algorithm T-RUEM when the number of different bids}
	\label{fig:exp2-T-RUEM}
\end{figure}

\begin{figure}[!htb]
	\centering
	\includegraphics[width=4.5in]{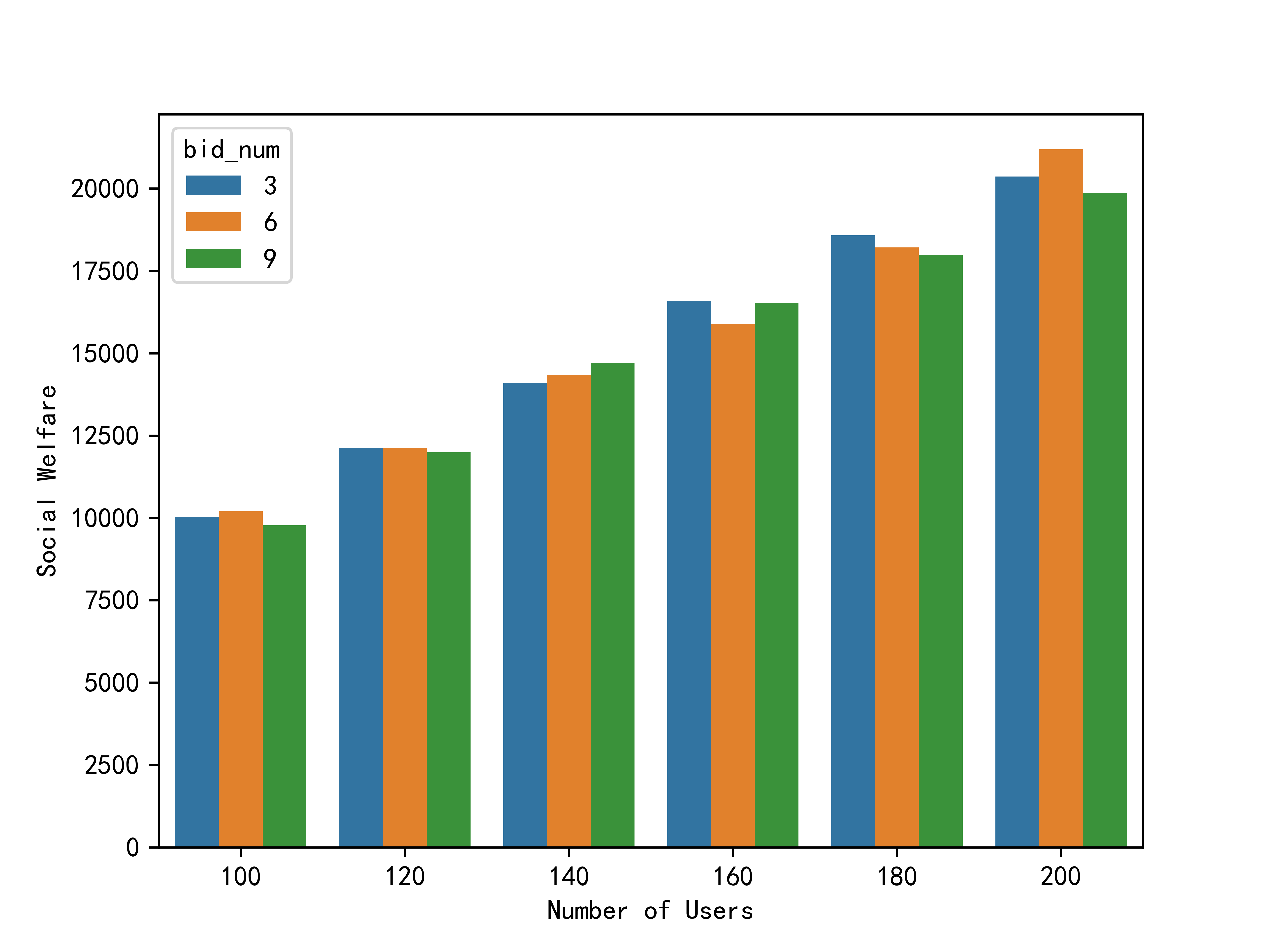}
	\caption{Social welfare comparison of algorithm T-RWAEM when the number of bids different bids}
	\label{fig:exp2-T-RWAEM}
\end{figure}

Experiment 3 aims to further analyze the situation under multiple bids. The purpose is to view the competitive ratio under the condition of multiple bids. Considering the completion within the valid time, the number of users is still set between 10 and 50, the resource types are fixed at 2 types, and the number of bids is 3, 6, 9 this time. The commercial solver gurobi is used to solve it. The experimental results are shown in Figure \ref{fig:exp3-compare2alg}. The experimental results show that the average competitive ratio of algorithm T-RWAEM is higher than that of algorithm T-RUEM. As the number of bids increases, the competitive ratio decreases slightly. This is because the increase in the number of bids indicates an increase in user requirements. The probability that the task can fall within the optimal time slot is relatively reduced compared to before, so the competitive ratio decreases.

\begin{figure}[htb]
	\centering
	\includegraphics[width=4.5in]{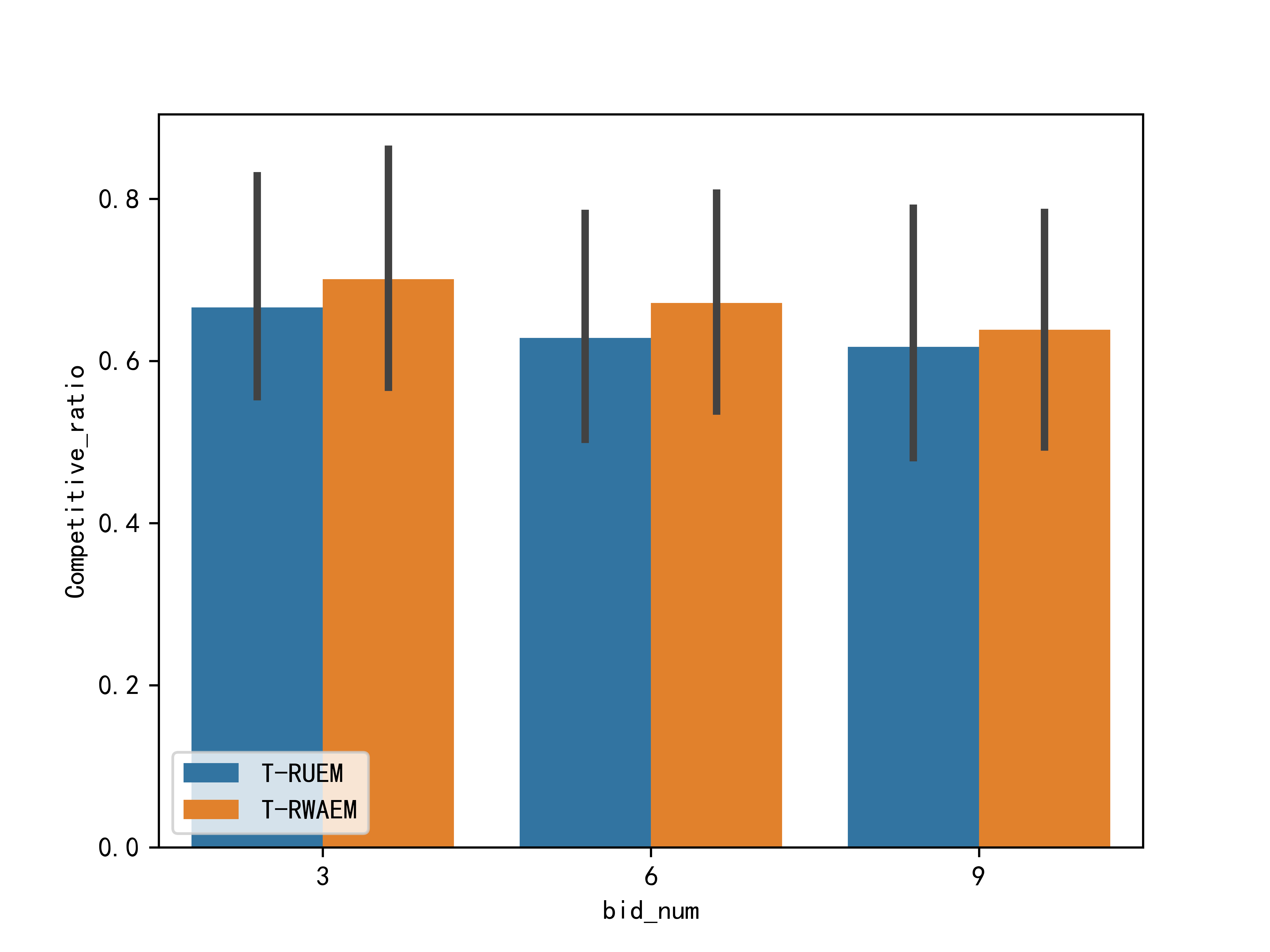}
	\caption{Competitive ratio between algorithm T-RUEM and algorithm T-RWAEM at different bid numbers}
	\label{fig:exp3-compare2alg}
\end{figure}

To discuss the adaptability of the two algorithms under multiple resources, Experiment 4 is designed. The experimental settings are: number of users is 100-200, number of bids is 3, number of resources is simulated data, also normalized to 1.
The experimental results are shown in Figure \ref{fig:exp4-total}, indicating that the proposed algorithms can still have good performance under multiple resources. Compared with the greedy benchmark algorithm, they both improve as the number of resource types increases. In addition, due to the increase in resource types, algorithm T-RUEM is more likely to reject tasks that may need to be kept for completion due to a single factor, indicating that algorithm T-RWAEM has a better adaptability. It can also be considered that T-RUEM is a special case of algorithm T-RWAEM under a single resource.

\begin{figure}[htb]
	\centering
	\includegraphics[width=4.5in]{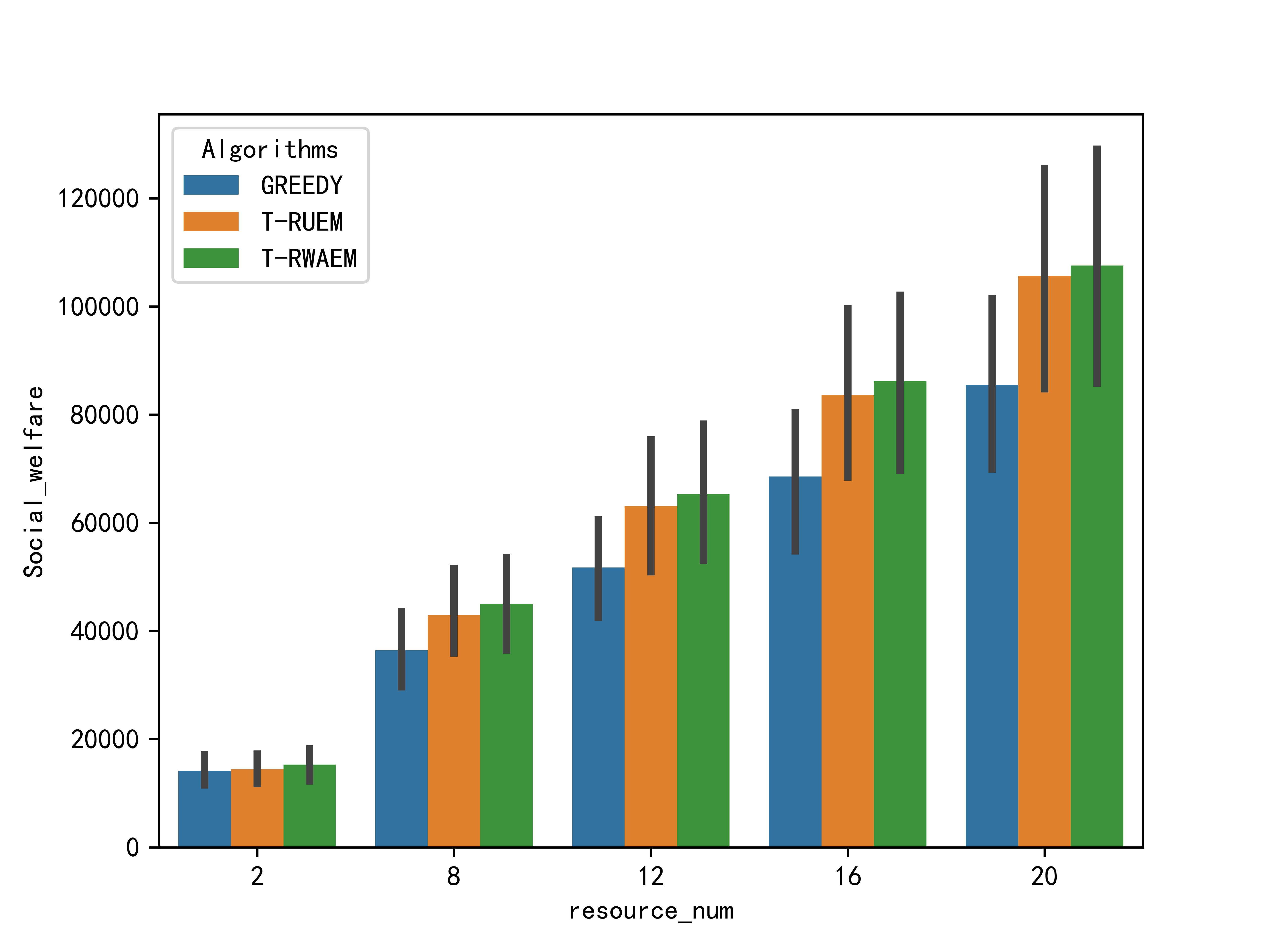}
	\caption{Total social welfare under varying number of resource types}
	\label{fig:exp4-total}
\end{figure}

\section{conclusion}\label{conclusion}

In this paper, we have presented a cloud auction framework with soft deadlines to address the online resource allocation problem. The main contributions are as follows:
1.We propose a solution using integer linear programming (ILP) modeling.
We design an auction time scheduling framework with three key modules - Allocation, Evaluation, and Action. This framework provides a systematic approach to dynamically allocate resources in a cyclic closed-loop manner.
2.We develop several evaluation methods, including time-based single resource utilization evaluation and weighted average evaluation, to assess the resource usage efficiency. These methods offer alternative optimization metrics for the scheduling framework.
We introduce a soft acceptance protocol that enables flexible online resource allocation. By temporarily overbooking but ensuring final feasibility, it improves resource utilization and user experience.
3.We analyze the time complexity of the proposed algorithms and prove them to be polynomial time, demonstrating efficiency for practical applications. The modular design also makes the framework easily extensible to incorporate more advanced algorithms.

In summary, this paper presents a structured cloud auction framework integrating optimization modeling, algorithm design, and online protocols. It provides useful insights and techniques for building practical cloud resource management systems. As future work, more complex models considering stochastic arrivals and multi-dimensional resource constraints could be investigated. The performance of the algorithms could also be evaluated on real-world cloud workloads. There are ample opportunities for further enhancements to the system's robustness, efficiency, and fairness.

\bibliographystyle{unsrt}
\bibliography{ref}

\end{document}